\documentclass[12pt]{amsart}
\usepackage{comment} 
\usepackage{amssymb}
\usepackage{amsmath}
\usepackage{amsfonts}
\usepackage{mathrsfs}
\usepackage{graphicx}
\usepackage{color}
\usepackage[onehalfspacing]{setspace}
\usepackage{caption}
\usepackage{enumerate}
\usepackage[round]{natbib}
\usepackage{appendix}
\usepackage{lscape}
\usepackage{subcaption}
\usepackage{graphicx}
\usepackage{amsfonts}
\usepackage{placeins}
\usepackage[utf8]{inputenc}
\usepackage{charter}
\usepackage[colorlinks=true,citecolor=blue,urlcolor=blue,pdfpagemode=UseNone,pdfstartview=FitH]{hyperref}
\usepackage{apptools}
\usepackage{multibib}
\usepackage{multirow}

\usepackage{dcolumn}
\usepackage{pdflscape}

\usepackage{booktabs}
\usepackage{threeparttable}

\newcites{main,supp}{References,References}
\makeatletter
\def\section{\@startsection{section}{1}
	\z@{1.0\linespacing\@plus\linespacing}{.8\linespacing}{\Large}}

\def\subsection{\@startsection{subsection}{2}
	\z@{.8\linespacing\@plus.7\linespacing}{.7\linespacing}{\large}}

\def\subsubsection{\@startsection{subsubsection}{3}
	\z@{.5\linespacing\@plus.7\linespacing}{-.5em}{\normalfont\bfseries}}
\makeatother

\numberwithin{equation}{section}

\newtheorem{proposition}{Proposition}[]

\newtheorem{lemma}{Lemma}[section]
\newtheorem{corollary}{Corollary}[section]
\theoremstyle{definition}

\theoremstyle{definition}

\theoremstyle{definition}

\title{}
\begin{document}
	\vspace*{2.5ex minus 1ex}
	\begin{center}
		\Large \textsc{When Predictions Become Regressors: A Split-Sample Correction for Biases in Downstream Inference}
		\bigskip
	\end{center}
	
	\date{%
		\today%
	}

\date{\today}
\vspace*{1.5ex minus 1ex}
	\begin{center}
		Nathan Canen and Ted Enamorado\\
	
	\medskip
	
	
	\end{center}
	
\thanks{
	\textbf{Canen:} (Corresponding Author) University of Warwick and CEPR. email: \url{nathan.canen@warwick.ac.uk} \\
	\textbf{Enamorado:} Washington University in St. Louis. email: \url{ted@wustl.edu}\\ 
We thank Shantanu Chadha for outstanding research assistance.}

	
\begin{abstract} 
Prediction-based methods, including Large Language Models (LLMs) and other machine learning techniques, are often used to construct measures of political phenomena that are difficult to quantify directly, such as policy positions in manifestos or emotions expressed on social media. In many applications, these prediction-generated measures are used as explanatory variables in regression models, even though they are measured with error. This leads to biased estimates. 
In this paper, we propose a simple solution to these biases: instrumental variables constructed from multiple measures created on independent splits of the original data. This approach is theoretically valid, easy to implement, and does not require new data. Through simulations, we show that this approach recovers estimates close to the true values, even in relatively small samples, while the standard approach can produce substantial bias in practice. We illustrate the method by revisiting two applications: whether gendered speech affects legislative outcomes in the German Parliament, and whether political risk influences poverty alleviation programs in China.

\end{abstract}
	\maketitle

\newpage
\section{Introduction}

The increasing availability of large datasets, whether from traditional or non-standard data sources (e.g., text, images, audio), has enabled political and social scientists to create new measures of objects that are hard to quantify, such as measures of sentiment or emotions (e.g., irony, anger) in text, policy positions from manifestos, measures of persuasion, media slant from news articles (at both aggregate and individual levels), among others. These measures are often obtained using prediction-based methods, such as Large-Language Models (LLMs), novel machine learning techniques, or even simpler methods such as logistic regression. However, such variables are rarely the target of empirical analyses per se. Rather, most applications construct these measures and then use them as explanatory variables in a linear regression model: e.g., to understand the role of emotions or slant on political or voting outcomes.\footnote{See \cite{linegar, li, coil} for extensive recent surveys, including of the many variables created with LLMs (e.g., Table 3 in \citealp{li}).}

In practice, most researchers use the generated variables ``as is" on the right-hand side of the regression, as if they captured the true, latent variable of interest. However, such prediction-generated variables are estimated with error. For example, an LLM's estimate of an article's slant or the emotional content of text may only be a proxy for the true value, as it uses a limited sample and subject to idiosyncratic (e.g., query-level) errors. This means that naively using these generated variables in a regression can generate endogeneity bias due to measurement error, as is well known from standard textbooks (e.g., \citealp{wooldridge}), and shown more recently when using such methods with unstructured data (\citealp{battaglia}) and LLMs (\citealp{duan}). 

The main contribution of this paper is to provide a solution for these biases that is (i) computationally simple, 
(ii) theoretically desirable, and (iii) does not require \textit{any more data or variables} than what the researcher used in the first place. The solution is based on creating a specific class of instrumental variables from repeatedly measuring the underlying concept on separate, independent splits of the original dataset. We then show the benefits of our approach using statistical results, simulations based on models found in the literature, and revisiting two empirical applications: one from \cite{ash} on whether gendered speech induces different outcomes in the German Parliament, and another on whether political risk (as evaluated from text) affects the implementation of government programs in China (\citealt{lin2025}).


To be more specific, our procedure uses predictions on different subsamples of the original data to create multiple measures for the same concept: e.g., multiple estimates for the slant of an article, or the emotional language in a passage. This is done through sample splitting: the same or multiple methods produce alternative measures based on each independent subsample. The multiple measures can then be used as instruments for one another: they are correlated because they measure the same underlying feature. However, their measurement errors are independent conditional on inputs (under standard assumptions in linear regressions) when they are independent errors from (conditionally) independent observations. We formally show this instrumental variable approach is valid (i.e., it satisfies the exclusion and exogeneity restrictions under the assumptions above), its associated estimator is consistent for the coefficients in the linear regression, we demonstrate finite-sample gains in simulations, and provide a simple algorithm that can be implemented using existing statistical packages. 

As a running example, one can think of an LLM that can be applied on each independent split, creating multiple measures for the same underlying concept. Yet, our results extend beyond LLM-based measures. Our main requirements are independent variability in measurement errors across observations measuring the same feature, and a sufficiently large dataset so that splitting is possible and the resulting subsamples are sufficiently large. With this approach, the researcher does not need a labelled subset (as in \citealp{angelopoulos}), to manually calibrate/classify some subset of that data (as in \citealp{duan}), to observe part of that data and validity at the limit (\citealp{yang22}) or assume the parametric distribution for the variable of interest, yielding more numerically cumbersome estimation (\citealp{battaglia}). On the other hand, there are certainly efficiency losses relative to using such additional information, and our approach requires measurement errors to be uncorrelated across measures, conditional on the input. We expand on these comparisons below.

We then illustrate our results in simulations that assess the extent of the measurement error bias and the finite-sample performance of our solutions. We provide multiple simulation designs, ranging from simple to more complex settings, mimicking the data structures and parameter values found empirically.\footnote{For example, we consider simulations where a document-term matrix representing word counts is generated, the measure of interest is constructed from the document-term matrix using an embedding estimator, which is then fed into a linear regression.} We also test a variety of measurement error structures. Across specifications, we show that the measurement error bias from using LLM-type generated regressors can be significant. 
For example, the OLS estimates are close to half of the true values, and other times they can be significantly above the true values (as we show theoretically). Meanwhile, the estimates based on our instruments are very close to the true values (e.g., often within 0.01) even with relatively small sample sizes for large datasets. 

Finally, we show that these conclusions hold in two applications. In the first, we revisit the empirical application of \cite{ash}. In their paper, the authors use a corpus of over half a million speeches from the German parliament and estimate a measure of gendered speech (i.e., whether a passage of text is more easily identified as a woman's speech than a man's) using a topic model. The authors use this estimated measure in a linear regression and show that speeches with language more closely associated with women's topics receive fewer reactions, and this effect is more pronounced when such speeches are delivered by men. We revisit this exercise with our instrumental variables estimator and show that these effects are significantly \textit{larger} in magnitude once measurement error is accounted for: for most outcomes, the effect of predicted gender on any interaction increases by 50-100\%. Thus, their original results may have significantly underestimated the true effects.

In addition, we revisit \cite{lin2025}. The paper examines how Chinese firms responded to the poverty alleviation campaign launched in 2015 and argues that firms facing greater political risk increased poverty-alleviation spending to strengthen ties with the state. To measure political risk, the author constructs a firm-level Political Risk Index from more than 418,000 investor-firm Q\&As, using BERT to classify whether each exchange is political. The author then incorporates this estimated measure into a linear regression (based on a difference-in-differences design) and finds that firms with higher political risk spent more on poverty alleviation after 2015. We revisit this analysis using our instrumental variables estimator and find that the effect remains positive and statistically significant but is, again, larger in magnitude. This finding is consistent with attenuation from measurement error in the machine-learning-generated regressor, suggesting that correcting for this error reveals an even stronger relationship between political risk and firms' strategic responses to government policy initiatives.

\section{Literature Review} 

The issue of measurement errors in generated regressors, particularly using LLM or machine learning methods, has been noted by various recent papers, such as \cite{angelopoulos, battaglia, duan, yang22, burtch}. \cite{ludwig}, in particular, provides an extensive an overview for LLMs, including associated evidence. In all of the former cases, researchers study the biases associated with generated regressors which are plugged-in to a linear regression. 

A main difference in our approach is that it does not require a ``validation set" of observations, where the underlying measure of interest is observed perfectly. Thus, we do not assume a researcher can observe (even a small) subset of the data without error (e.g., as in \citealp{angelopoulos, duan, yang22, burtch}). In particular, closest to our work are \cite{yang22, burtch} who propose the use of instrumental variables to correct for the endogeneity bias from the measurement error from generated regressors. However, in both papers, the authors assume that they can observe a (small) subset of the data that is correctly labeled. This allows them to construct instruments leveraging the observable error to the unlabeled data, whether through random forests (\citealp{yang22}) or ensemble methods (\citealp{burtch}).\footnote{For example, \cite{yang22} shows that the prediction errors across trees will become uncorrelated (as sample size grows), so they can be used as instruments for one another.} Of course, if such information is available, then it would be more efficient to use it. \cite{duan}, for example, proposes a bias-correction estimator that combines such calibrated (labeled) data into the original regression. Meanwhile, as they can observe the measurement error, \cite{burtch} do not require independence across measurement errors (conditional on the input data). Another alternative would be to use a design-based approach and extrapolate the measurement error from the differences between predicted and actual outcomes for a subset of the data (e.g., \citealp{egami, rister}). This requires a focus on a binary treatment and on a different asymptotic approach. However, the availability of validation data may fail in many settings, as pointed out by \cite{battaglia}. Indeed, in political science, it is unlikely that researchers could have full confidence in exact media slants, or emotions in political discourse. Our proposal may be desirable in these settings.\footnote{While\cite{ludwig} suggest that when estimating an LLM without validation data and a ground truth exists, that the ``practical answer is clear: invest effort and collect a small validation sample", this may not be feasible.} 

Our solution is similar in spirit to \cite{gillen}'s approach to dealing with measurement error in experimental data: in their case, they suggest running the same experiment multiple times and using later experimental outcomes as instruments for the first. This is a valid instrument if the measurement errors in the first experimental run are independent of the latter errors. In our case, the researcher does not even need to run new experiments that can be costly. Rather, when the original dataset is large enough, the experiment is ``mimicked" through creating independent sample splits (i.e., input samples) and the LLM outputs the multiple measures without additional collection. In this way, it builds on a tradition of multiple measures trying to replicate instruments (e.g., using twins as in \citealp{ashenfelter}). The resulting estimator is an instrumental variables estimator with two samples (see \citealp{angrist92,inoue} for instance), but where the samples are generated through random splits of the data rather than multiple datasets.

Our multiple measures solution does not require assuming a specific parametric form of the type of measurement error that the researcher faces, as in \cite{battaglia}. This is useful when the researcher is unwilling to take a stance on the class of models generating the regressor (e.g., coming from a multinomial logistic, as in the former; or from a specific type of neural network and applying a result like \citealp{farrell}) as different first-step estimators generate different expressions for bias). It is also less computationally costly. On the other hand, if such parametric forms are deemed reasonable,  \cite{battaglia} derives the expression of the bias of the estimand, proposes a more efficient estimator that can remove this bias and, thus, may have significant statistical gains.

A main limitation of our approach, however, is that we assume that the errors across measures \textit{conditional on the input (thus, latent feature) and other observables} are uncorrelated across observations. This is also used in \cite{yang22, burtch}, for example. Simulations show that, when this holds, our results work very well with small samples. However, correlations induced by unobservables are plausible in many settings, whether due to the way predictions are constructed, data structures (e.g., network, geographical data that is unaccounted for in prediction), among others.


 \section{Set-Up}\label{main}

Let $\{(Y_i, Z_i)\}_{i=1}^n$ be i.i.d. from a representative population, where $Y_i$ represents the outcome for individual $i$ (e.g., voting behaviour, approval of a politician, etc.), and $Z_i \in \mathbb{R}^K$ represents the vector of features that will be used to construct the measure of interest (e.g., words from a speech, the words in a media article, etc.).

The measure of interest is defined by 
\begin{eqnarray}
X_i = f(Z_i),\label{Xi}
\end{eqnarray}
where $f(\cdot)$ can represent an LLM (e.g., ChatGPT), another machine learning model (e.g., clustering), a known function (e.g., taking the population mean of those variables), among others. The measure $X_i \in \mathbb{R}$ is one-dimensional (e.g., sentiment of a post $i$, gendered speech, the topic(s) in document $i$, etc.). 

Our target model is the linear regression model:
\begin{eqnarray}
Y_i = \beta_0 + \beta_1 X_i + \varepsilon_i,\label{eq1}
\end{eqnarray}
where $\mathbb{E}[\varepsilon_i \mid X_i]=0$, so that the error term is exogenous to the true measure of interest (e.g., a measure of emotional content, as in \cite{gennaro}, or slant of one's feed, as in \citealp{braghieri}). This assumption isolates the role of measurement error in driving endogeneity from ML-driven measurement error. It is also commonplace in this setting, such as \cite{yang22, burtch}. 

We present the simple version of equation (\ref{eq1}), with no other covariates, for expositional simplicity. Without loss, we can always residualize the other covariates, and instrument, and obtain an analogous expression.

Unfortunately, estimating (\ref{eq1}) is infeasible, as we can only measure $X_i$ with error. Thus, we can use our sample $Z_i$ to produce an estimated measure of $X_i$, denoted $\hat{X}_i$.

\medskip

\textbf{Assumption 1: (Measurement Error Structure)}: We assume that we can only measure $X_i$ with error, following 
\begin{eqnarray}
\hat{X}_i = X_i + \eta_i,\label{eq2}
\end{eqnarray}

where $\eta_i$ is the measurement error and satisfies:
\begin{enumerate}[(i)]
\item 
 $\{(Z_i,\eta_i,\varepsilon_i)\}_{i=1}^n$ are independent across $i$ in the full sample.
\item
$\mathbb{E}[\eta_i \mid Z_i] = 0$ 
\end{enumerate} 

\bigskip

We can always rewrite the estimated value $\hat{X}_i$ as the true value $X_i$ plus an error term. Thus, the main assumption in (\ref{eq2}) is not the additive structure or that it is mean zero. It is mean-zero as the constant can be captured in the $\beta_0$ term. Rather, it is that errors are independent across $i$. This is weaker than classical measurement error in the sense that it does not assume a constant variance for $\eta_i$ - see the examples below. Our main propositions are also robust to within-unit correlation between $\eta_i, \varepsilon_i$.
Throughout the paper, we assume that $Var(X_i)>0$, and the second-moments of $X_i, \varepsilon_i, \eta_i$ are bounded. We now discuss this main assumption.

\subsection{The Interpretation of Assumption 1}

Assumption 1 implies that our observed measure $\hat{X}_i$, differs from the true value, $X_i$, through measurement error $\eta_i$. We can interpret this measurement error through multiple possible sources in machine learning (ML)-type methods. 

In particular, suppose that our observed measure can be written as:
\begin{eqnarray}
\hat{X}_i = \hat{f}(Z_i)+\xi_i,\label{Xihat}
\end{eqnarray}
where $\xi_i$ is independent across $i$ (conditional on $Z_i$).

This error structure can capture two separate sources of measurement error that are inherent to ML-type methods. First, such methods are based on pretrained models that use other datasets. Thus, we may be only able to observe what comes from an estimated model $\hat{f}(Z_i)$, rather than its population (``true") counterpart, $f(Z_i)$ where $\hat{f}(\cdot)$ is a \textit{pretrained} estimator (i.e., not a function of the sample $\{(Y_i, Z_i)\}_{i=1}^n$). Second, there may be a different source of error when computing $\hat{X}_i$ rather than $X_i$ (i.e., measurement noise $\xi_i$). This may be due to the sampling variability, or the randomness in the output of a measure (such as due to random steps in outputs, queries, lack of seeding, idiosyncracies such as the effects of small changes to prompt wording), etc.

Subtracting \eqref{Xi} from \eqref{Xihat}, we can then obtain:
\begin{eqnarray}
\hat{X}_i - X_i = \underbrace{\hat{f}(Z_i) - f(Z_i)}_{model~error} + \underbrace{\xi_i}_{measurement~noise},\label{etai}
\end{eqnarray}
so that $\eta_i = \hat{f}(Z_i) - f(Z_i) + \xi_i$.

Assumption 1 requires that the first term is zero, which can be interpreted as $f(z) = \mathbb{E}[\hat{X}_i \mid Z_i = z]$ being the population counterpart to the pre-trained model. For example, the researcher is interested in the output from a ``true" association between $Z_i$ and $X_i$, they must settle for the pre-trained (estimated) version from available models from ChatGPT or Claude (with some idiosyncratic output).  Thus, $\beta_1$ is the parameter on the model's population measure. In this set-up, $\eta_i = \xi_i$ and our results show how to correct for the idiosyncratic variation in a model's output, even when the model used is the one of interest. 


Note that we require independent variation in $\eta_i$ and in $\hat{X}_i$ across $i$, even conditional on $Z_i$. This comes through variability in $\xi_i$. Crucially, the independence of $\xi_i$ adds independent variation across $i$ onto the observed measure $\hat{X}_i$. This can rationalize how inputs $Z_i = Z_j$ generate different observed measures $\hat{X}_i \neq \hat{X}_j$. Thus, even if $Z_i = Z_j$ may end up generating the same $X_i$, the observed variables may be different even conditional on $Z_i$.


Examples 1 and 2 below present cases that can fit the above set-up.

\bigskip

\textbf{Example 1:} If $\eta_i$ is classical measurement error (i.e., mean $0$, variance equal to some $\sigma^2_{\eta}>0$), then the estimand for the parameter for $\hat{X}_i$ from equation (\ref{eq3}) equals $\beta_1 \frac{\sigma^2_{X}}{\sigma^2_{\eta} + \sigma^2_X}$, where $\sigma^2_X = Var(X_i)$. This is the classical attenuation bias case. See \cite{wooldridge}, for instance.

\bigskip

\textbf{Example 2:} Let $Z_i \sim Binomial(C_i, X_i)$. For example, $Z_i$ is the number of articles containing specific terms $i$, $C_i$ is the total number of articles that month, and we wish to know $X_i$, the true measure of uncertainty (in \citealp{baker}) or slant (\citealp{braghieri}).  Then, \cite{battaglia} show that, under an asymptotic sequence $\sqrt{n} \mathbb{E}\left(\frac{1}{C_i}\right) \to \kappa \in [0, \infty)$, then $\sqrt{n}(\hat{\beta}_1 - \beta_1)$ has asymptotic bias $-\kappa \beta_1 \frac{\mathbb{E}X_i(1-X_i)}{Var(X_i)}$. This argument generalizes to $\hat{f}$ of multinomial logistic form, including \cite{gentzkow} among others.

\subsection{The Feasible Model}

Under Assumption 1, we can rewrite the true model as a feasible counterpart:
\begin{eqnarray}
Y_i &=& \beta_0 + \beta_1 \hat{X}_i + \varepsilon_i - \beta_1(\hat{X}_i-X_i)\notag \\
&=& \beta_0 + \beta_1 \hat{X}_i + u_i,\label{eq3}
\end{eqnarray}
where $u_i = \varepsilon_i - \beta_1\eta_i$, and the first line adds and subtracts $\beta_1 \hat{X}_i$. As $u_i$ is correlated with $\hat{X}_i$ (both are functions of $\eta_i$, by Assumption 1), there is endogeneity. This is the argument in standard statistical textbooks (e.g., \citealp{wooldridge}).

However, there is a key difference in our setting relative to textbook references: the form of the bias may differ from classical attenuation bias. We formalize this in Lemma \ref{lemma1} in Appendix. This is because we do not assume $Cov(\eta_i, \varepsilon_i)=0$: we allow for dependence within-unit, because our instrument's validity only depends on the covariance between $\eta_{nn(i)}, \varepsilon_i$ which we will show below is equal to 0.

\subsection{Our Solution}

Rather than requiring further data classifications (e.g., \citealp{duan}), estimating complex likelihoods (\citealp{battaglia}), or not accounting for frequentist inference (\citealp{egami}), we propose a very simple, tractable approach, easy to implement, and with an associated consistent estimator. In summary, we propose using instrumental variables. But our instrumental variables can be generated from the existing dataset, so, under Assumption 1, the typical concerns about finding a valid one do not apply.

To create such instrumental variables, we combine the commonly used sample splitting approach in machine learning with the intuition from experimental social science where a researcher uses additional experimental measures as instruments for each other (e.g., \citealp{gillen}). In essence, by splitting the sample, we can create multiple (noisy) measures for $X_i$. As the measurement error is independent across observations, and the split is independent, any correlation across those measures is solely due to the latent (and true), $X_i$. Thus, they are correlated because they measure the same latent object. But they are now exogenous, because the measurement errors are uncorrelated.

To do so, take our original dataset $\{(Y_i, Z_i)\}_{i=1,...,n}$. Now, split that dataset into two independently drawn datasets so that each observation belongs to only one dataset. Call them the \textbf{training} and \textbf{instrument} datasets. As the split is random, the distribution of $\{(Y_i, Z_i)\}$ is the same across both datasets. 

Thus, for each target observation $X_i$, we have two measures: $\hat{X}_i$, obtained from applying the prediction-based method $\hat{f}(\cdot)$ (whether LLM, or another technique) to data from the first independent split, and $\hat{X}_i^{instr}$, obtained from applying the method to the second, where the superscript $instr$ refers to an observation from the instrument data.\footnote{Naturally, we can have many more splits than two. The arguments are analogous. Then, the researcher has the choice of how to construct instruments from the many other measures - e.g., by aggregating them, averaging, etc.} Crucially, we require that the data split into training and instrument datasets, and the matching rule (assigning each $i$ to its matching observation) cannot depend on outcomes or other error terms. As we will see, this independence guarantees that both measures capture part of the latent variable, but their errors are still independent. This is summarized in Assumption 2.

\medskip

\textbf{Assumption 2: (Randomization in Sample Splitting)} The data split into training and instrument datasets follows a randomization $U$ that is independent of $\{Z_i,\eta_i, \varepsilon_i\}_{i=1}^n$. Furthermore, the matching rule is a function of $\{Z_i\}_{i=1}^n$ and this randomization $U$ only.

\medskip

Note that, as the observation $i$ is only in the training sample, $\hat{X}^{instr}_i$ refers to the measure for $X_i$ computed using some other (appropriate) observation in the instrumental data. 

Using the notation in \eqref{Xihat}, we have that:
\begin{eqnarray}
\hat{X}_i &=& \hat{f}(Z_i) + \xi_i, \\
\hat{X}_i^{instr} &=& \hat{f}(Z_{nn(i)}) + \xi_{nn(i)},
\end{eqnarray}
for some observation $nn(i)$ in the instrument data. Thus, even if we use inputs $Z_{nn(i)} = Z_i$ and the same pre-trained model $\hat{f}(\cdot)$, $\hat{X}_i$ and $\hat{X}^{instr}_i$ will still differ through the random variables, $\xi$'s.  The next section suggests approaches to choose $nn(i)$ efficiently, such as selecting the nearest neighbor of $i$, in finite samples.

We propose to use the new measure $\hat{X}_i^{instr}$ - based on the instrumental data - as an instrument for $\hat{X}_i^{train}$. The main theoretical contribution of this paper is to show that this instrument is valid: i.e., it is exogenous $Cov(\hat{X}_i^{instr}, u_i)=0$ where $u_i = \varepsilon_i - \beta_1\eta_i$ from (\ref{eq3}), and  relevant (i.e., that $\hat{X}_i^{instr}$ is correlated with $\hat{X}_i$).\footnote{Conceptually, one could allow $\hat{f}$ to differ across each measure, thereby representing different models. The arguments remain the same, since the differences in those models are fixed (conditional on $Z$).} The theoretical result is summarized in Proposition \ref{proposition} below, with the proof in Appendix \ref{proofs}. The next sections then show the empirical performance of this solution.

\begin{proposition}\label{proposition}
Let $(\hat{X}_i, \hat{X}_i^{instr})$ be two separate prediction-generated regressors computed from different units, from separate independent splits of the data, with perfect matching $Z_{nn(i)}=Z_i$. 

Then, under Assumptions 1-2, $Cov(\hat{X}_i, \hat{X}_i^{instr}) \neq 0$ and $Cov(u_i, \hat{X}_i^{instr}) = 0$.

Thus, $\hat{X}_i^{instr}$ is a valid instrument for $\hat{X}_i$.
\end{proposition}

Intuitively, the first part of Proposition \ref{proposition} holds because both variables measure the same latent variable, $X_i$ - they just do so with different errors. (As we will see below, this is also true even when using different inputs, $Z_i$). Thus, they are correlated. The crux of the argument is about exogeneity: that $\hat{X}_i^{instr}$ is uncorrelated with the error from (\ref{eq3}) with only the training data, $u_i^{train} = \varepsilon_i^{train} - \beta_1 \eta_i^{train}$, where the superscript $train$ refers to an observation in the training data. Indeed, while both $\hat{X}_i^{instr}$ and $\hat{X}_i^{train}$ measure $X_i$ with error, their errors ($\eta_i^{train}, \eta_i^{instr})'$ are uncorrelated by Assumption 1 and the randomization in Assumption 2. This is related to \cite{gillen}, who assume a related condition across different experiments. Here, rather than running additional experiments (for the same $i$), our experiments are deriving measures from other splits of the data (or independent LLMs).

\subsubsection{Robustness of the Result to Imperfect Matches}

Proposition \ref{proposition} assumed that we can find pairs $(Z_i, Z_{nn(i)})$, each from a different sample, such that $Z_i = Z_{nn(i)}$. As the two samples have the same underlying distribution, in large samples and when $Z_i$ is discrete, this should occur with positive probability. And when $Z_i$ is continuous, we will be finding matches that can be arbitrarily close in large samples, although never exactly true. Furthermore, in small samples, this ``perfect match" assumption may fail, particularly when $Z_i$ is continuous. 

To deal with this issue, we could find alternative $Z_{nn(i)}$ to $Z_i$, such as some that are close (although not equal). The algorithm in the next subsection proposes ways to do so, such as nearest neighbor matching - see (Step 2). 

And, Proposition \ref{prop_imperfect} in Appendix shows that Proposition \ref{proposition} holds almost exactly under imperfect matching: the exogeneity of the instrument is not affected by the imperfect match (under the stronger exogeneity of $\varepsilon_i$ to $\mathbb{E}[\varepsilon_i \mid Z_i]=0$). In fact, match quality only affects first-stage strength, but not the consistency of the estimator: the instrument will be relevant as long as the $Cov(X_i, X_{nn(i)}) \neq 0$ - i.e., that the chosen match has some (even if small) correlation. This also implies that a worse match will affect the strength of the instrument. 

\subsubsection{Estimation}

Naturally, once we consider a linear regression model (\ref{eq1}) with endogeneity, and have provided a valid instrument (Proposition \ref{proposition} and \ref{prop_imperfect}), then $\beta_1$ can be consistently estimated using a standard instrumental variables (IV) estimator with textbook inference. To invoke standard asymptotics, however, we requires that no observation in the instrument data is used as a match for more than one observation in the training data (i.e., unique matches). This guarantees a lack of dependence across the training data observations (i.e., the same instrument observation is not used many times, generating additional correlation across such observations). 

With this condition, we are back to the setting with an IV estimator with i.i.d. data, with an instrument which satisfies exclusion and relevance conditions. This is summarized in the corollary below. With many instruments, one can use Two Stage Least Squares (2SLS), as long as its regularity conditions are satisfied accordingly.

\medskip

\begin{corollary}
Suppose the conditions of Proposition \ref{prop_imperfect} hold and that the matching rule is such that each observation in the instrument data is matched to one observation in the training data. Then, the instrumental variables estimator from using $\hat{X}_i^{instr}$ as an instrument for $\hat{X}_i$, denoted $\hat{\beta}_{IV} = \frac{\sum_{i=1}^{n_1} \left(\hat{X}_i^{instr}-\overline{\hat{X}^{instr}}\right) (Y_i - \overline{Y})}{\sum_{i=1}^{n_1} \left(\hat{X}_i^{instr}-\overline{\hat{X}^{instr}}\right)(\hat{X}_i - \overline{\hat{X}})}$ is consistent for $\beta_1$, where $n_1$ is the number of observations in the training data, $n_2$ the number of observations in the instrument data, and $n_1 \to \infty, n_2 \to \infty$. 
\end{corollary}

\bigskip

\subsection{Practical Implementation: An Algorithm}

In practice, one may not have the exact same-valued $Z_i$'s (or texts that generate the same $X_i$'s) in multiple splits of the data, so that they may not be measuring the same $X_i$. This can be addressed by finding the closest possible values across the two datasets. 
We propose using nearest-neighbor matching (based on $Z_i$) to find a match for the training observation in the instrument dataset. Thus, we now have two measures for each observation in the training dataset: their original measure from the training data, $\hat{X}_i^{train}$, and the measure based on the nearest neighbor $Z_{nn(i)}$ given by $\hat{X}_i^{instr}$ where $nn(i)$ is an index for the nearest neighbor for $i$, some $j \neq i$. They will be correlated as $f(Z_i)$ and $f(Z_{nn(i)})$ will be close and measure (approximately) the same $X_i$. This is summarized in the following algorithm.

\bigskip

\noindent\rule{16cm}{1pt}

\textsc{An IV Approach to Correcting Biases from Prediction-Generated Regressors}

\quad \textbf{Input:} 
\begin{enumerate}[(i)]
\item
Variables $Z_i$ (used to construct the target measure)
\item
Outcomes $Y_i$.
\item
Chosen method to generate $X_i$ from $Z_i$ (e.g., a specific LLM).
\end{enumerate}

\medskip

\quad \quad \quad (\textbf{Step 1}) Independently split the data $\{(Y_i, Z_i)\}_{i=1}^n$ into two groups. Call them the training data and the instrument data. Each subset should, thus, have the same distribution.

\medskip

\quad \quad \quad (\textbf{Step 2}) For each $i$ in the training data, find the $j$ in the instrument dataset that is most similar to $i$. For example, via nearest-neighbor matching.\footnote{For our main theoretical results, we assume that each data point is matched to a single (different) match.} 

\medskip

\quad \quad \quad (\textbf{Step 3})  Then, for each $i$, generate the intended covariate using the original $Z_i$ and the associated $Z_{nn(i)}$ from the instrumental data. The output is $\hat{X}_i$ and $\hat{X}_i^{instr}$.

\medskip

\quad \quad \quad (\textbf{Step 4}) Estimate $\beta_1$ by instrumental variables, using $\hat{X}^{instr}$ as an instrument for $\hat{X}_i$ in the regression using the training sample: 
\begin{eqnarray}
Y_i &=& \beta_0 + \beta_1 \hat{X}_i + u_i.
\end{eqnarray} 
Both the estimate and standard errors are obtained using the standard (closed-form) formulas for IV estimators - see \cite{wooldridge}. If there are many instruments, use Two Stage Least Squares (2SLS) accordingly. \\
\noindent\rule{16cm}{1pt}

\subsection{Discussion}

\subsubsection{The Independence Assumption on $\eta_i$}

Our key assumption is that the measurement error terms, $\eta_i$, are independent across $i$. This allows us to use the two measures, $\hat{X}_i$ and $\hat{X}_i^{instr}$. This assumption rests on independent variation across these measures, even conditional on $Z_i$. (In the notation of \eqref{Xihat}, this comes through independent variation in $\xi_i$).  This is a commonplace assumption (e.g., independent errors in experiments in \citealp{gillen}). It also seems more defensible when a researcher uses measures where the errors are likely to be independent (either due to randomness inherent to machine learning or LLM outputs, or through researcher inputs, e.g., typos or prompt errors), or when the researcher uses measures from different models (whose errors are likely to be different than those in the different subsamples). However, the assumption may fail in empirically relevant scenarios. This includes when (i) the researcher uses the same pre-trained model for both measures and they are not independent draws from the model (i.e., there are correlated errors violating Assumption 1(ii)), or when (ii) spillovers lead to correlation in the measures. We illustrate the effect of such violations on our results through some simulation designs in Section \ref{MC}.

\subsubsection{Efficiency Gains}

The benchmark Instrumental Variables estimator only uses half (or less) of the outcome data: those in the training sample. However, one could follow \cite{gillen} and simply swap the labels, obtain two IV estimators, and average them. As they discuss, this should be more efficient as it uses all available data, without requiring any further assumptions (Assumption 1 and 2 apply to all the data, not only to the observations in the training data). A similar idea can be implemented when many more measures (instruments) are available.

\section{Monte Carlo Simulations}\label{MC}

To illustrate the extent of the bias in the traditional approach (using the prediction-generated regressor ``as is" in a linear regression), and to showcase the finite-sample performance of our instrumental variable estimator, we conduct extensive Monte Carlo simulations. 
We keep the simulation design as close as possible to empirical contexts and to the empirical applications presented below. 

Overall, our simulation results highlight the extensive bias in OLS estimates. These simulations also show that our approach performs very well across parameter values and designs. In fact, our estimator is very close to true values with a much smaller sample size than is typically used in practice.

\subsection{Overview of the Simulation Design}

In these simulations, each observation corresponds to a document indexed by $i = 1,\dots, n$. Rather than assuming we observe a low-dimensional latent variable $X_{i}$ produced by a prediction model (e.g., an LLM), we construct a more realistic high-dimensional environment that mirrors how these methods process text (see \citealp{bosl:etal:2025} for a similar design). 
In particular, each document is represented by a document–term matrix (DTM) of word counts and the corresponding embedding. From these embeddings, we generate a latent label and its noisy prediction, while introducing this noise in a controlled fashion. 
This design allows us to examine how, e.g., commonly used modern text-as-data methodologies create structured errors that propagate into downstream regression analyses. 

\subsubsection{Simulation Set-Up}

As described above, for each document $i$, we simulate a high-dimensional vector of token frequencies. Let $Z^\top_i = (Z_{i1}, \dots, Z_{iK})$ be the vector of counts across $K$ vocabulary terms. Word counts are generated via a Poisson model with heterogeneous frequencies, where  $Z_{ik} \sim \text{Poisson}(\lambda_k)$ and $\lambda_k \sim \text{Gamma}(\theta, \psi)$, with $\theta = 1.2$ and $\psi = 0.3$.\footnote{The parameters are chosen so that the marginal distribution has a mean-to-variance ratio consistent with empirical word frequency distributions, i.e., sparse and counts concentrated in a few words \citep{pian:2014}. In Table~\ref{tab:robustness.sims} we show that our results are robust to different values of $\theta$ and $\psi$.} This produces sparse, heavy-tailed DTMs as commonly found in text corpora (\citealp{tadd:2015, inou:etal:2017}).

Next, we map each count vector $Z^\top_i$ to a $D$-dimensional embedding $E^\top_i$, mimicking how LLMs transform raw text into dense semantic vectors (\citealp{devl:etal:2019}). To that end, we define $E^\top_i = \tanh\big(Z^\top_i W \big)$, where $W$ is a $K \times D$ weight matrix, and each $W_{kd} \sim \mathcal N\left(0, \frac{1}{K}\right)$. Note that $\tanh(\cdot)$ implements a mild nonlinearity similar to activation functions used in neural networks (see e.g., \citealp{lecu:etal:1998}).\footnote{In Table~\ref{tab:robustness.sims} we show that our results are robust to common activation functions such as the Rectified Linear Unit (ReLU), with $E_{id} = \max (0, Z_{i} W_{\cdot d})$} This step induces shared document structure (through $Z^\top_i$), nonlinear transformations (through $\tanh(\cdot)$), and co-movement across embedding dimensions (through $W$).

We assume that documents possess an underlying true value $X_i$ which is a logistic transformation of a sparse linear combination of the embedding coordinates, i.e., $p_i = E^\top_i \delta$, with $X_i = \text{logit}^{-1}(p_i)$. In this setting, the $D$-dimensional vector $\delta$ is sparse, with $n_{\text{signal}}$ non-zero entries, which means that only a few coordinates carry information.

We generate noisy predicted probabilities as follows:
$$\widehat{X}_i = \text{logit}^{-1}(p_i + \eta_i), \qquad \tilde{\eta}_i \overset{\text{iid}}{\sim}  \mathcal{N}(0, \sigma_{\eta}^2).$$


We generate outcomes using the following equation:
$$Y_i = \beta_0 + \beta_1 X_i + \varepsilon_i, \qquad \varepsilon_i \overset{\text{iid}}{\sim}  \mathcal N(0,1),$$

with $\beta_0 = 0.5$ and $\beta_1 = -2$ as the parameters of interest. Because the researcher only observes $\widehat{X}_i$ rather than $X_i$, the OLS estimates from a regression of $Y$ on $\widehat{X}_i$ suffer from bias due to measurement error introduced in the embedding-driven prediction step.\footnote{We note that $\eta_i = \hat{X}_i - X_i$ differs from $\tilde{\eta}_i$ which is the error drawn within the nonlinear model. However, the violations are very mild, under 1\% in the Low Noise case, and the design and closeness to the assumptions validate the design.}

As described above, to correct for measurement error, we construct an instrument using documents that are close in the embedding space. 
To do so, we randomly split the sample into 1) a training set and 2) an instrument set. For each training observation $i$, we find its 
nearest neighbor in the instrument set using Euclidean distance ($L_2$ norm) in the embedding space ($E$):
 $$j(i) = \arg\min_j L_2(E_i, E_j ) $$

The key intuition is simple: documents close in embedding space share semantic structure, but prediction errors for different documents 
remain conditionally independent. As a result, nearest-neighbor predictions serve as valid instruments that are correlated with 
the true signal but uncorrelated with individual prediction noise. Once a match has been found, 
we use $\widehat{X}_i$ as the noisy regressor and $\widehat{X}_{nn(i)}$ as the instrument. 

Using this instrument, we estimate $\beta_1$ with the instrumental variables: 
$\widehat{\beta}_1^{\text{IV}}$. We then compare these results to the biased OLS estimator using noisy predictions and
the oracle estimator using unobserved true probabilities. 

Across simulations, we vary the number of documents $n \in \{5000, 20000, 50000 \}$. We set the number of
documents used to train the prediction model to $n_{\textrm{labeled}} = 1250$, 
the number of features $K$ is set to 2500 tokens, and the number of dimensions for each embedding is set to 100. Finally, we vary $\sigma_\eta \in \{0.25, 1, 2\}$.
In each specification, a neural network is trained to predict $X_i$ from the embeddings; 
each scenario is repeated 500 times.

\subsubsection{Simulation Results}

Figure~\ref{fig:fig01} presents the distribution of $\widehat{\beta}_1$ across 500 replications, faceted by sample size ($n = 5000, 20000, 50000$) 
and stratified by noise level, comparing OLS on $\widehat{X}_i$ to our nearest-neighbor IV estimator. 
Two patterns stand out. First, OLS estimates are persistently attenuated toward zero relative to the true value $\beta_1=-2$ (dashed line), 
and this attenuation barely moves with sample size: the point estimates in the $n = 5000$ and $n = 50000$ panels sit at roughly the same 
distance from the truth. This is consistent with a fixed asymptotic bias, i.e., increasing $n$ shrinks sampling variance, but it does nothing 
to correct the bias induced by (heteroscedastic) measurement error in $\widehat{X}_i$. Second, IV estimates are centered on the true value in every panel, 
and the percentile bands visibly tighten as $n$ grows, indicating that the instrument is consistent. The gap between 
the two estimators widens with the noise level: under Low noise the two intervals nearly overlap, whereas under High noise the OLS 
estimates collapse toward zero while IV remains anchored near the truth.

\begin{figure}[htbp]
  \centering
  \caption{Distribution of $\widehat{\beta}_1$ across 500 Monte Carlo simulations, by estimator (OLS on $\widehat{X}_i$ vs. IV with nearest-neighbor instrument)}
  \includegraphics[width=0.85\linewidth]{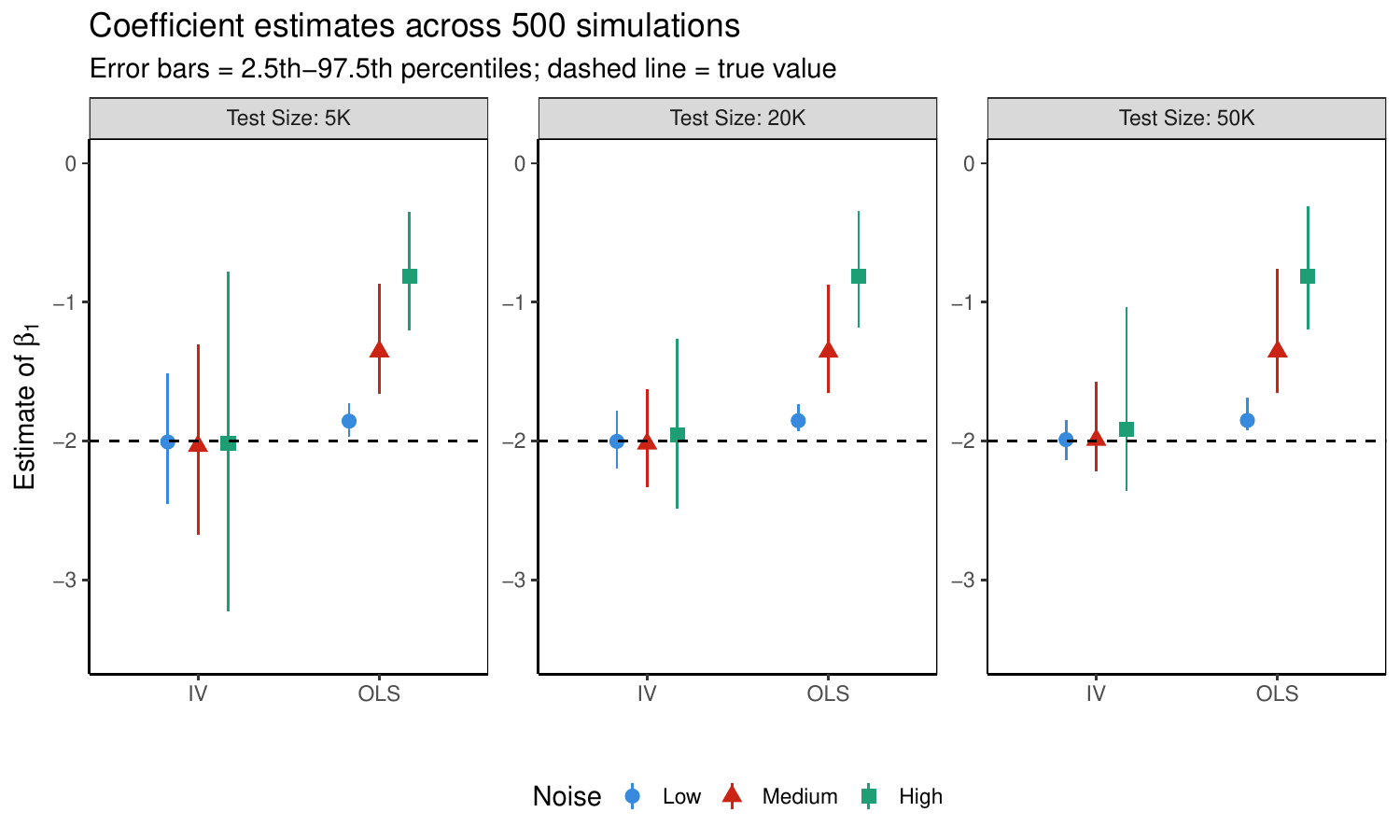} 
  \label{fig:fig01}
	\parbox{6.4in}{\footnotesize Notes: Sample size ($n=5000, 20000, 50000$) and noise level (Low: $\sigma_\eta=0.25$; Medium: $\sigma_\eta=1$; High: $\sigma_\eta=2$). Vertical bars span the 2.5th and 97.5th percentile of simulated estimates; the horizontal dashed line marks the true value ($\beta_1=-2$). OLS estimates are attenuated toward zero in every cell and do not converge to the truth as $n$ grows, reflecting a sample-size-invariant asymptotic bias; IV estimates are centered near the true value and their dispersion shrinks with $n$.}
\end{figure}

\begin{table}[!h]
\centering
\caption{Root mean squared error (RMSE) of $\widehat{\beta}_1$ across 500 simulations, by estimator, sample size, and noise level}
\centering
\begin{tabular}[t]{lcccccc}
\toprule
\multicolumn{1}{c}{ } & \multicolumn{2}{c}{5K} & \multicolumn{2}{c}{20K} & \multicolumn{2}{c}{50K} \\
\cmidrule(l{3pt}r{3pt}){2-3} \cmidrule(l{3pt}r{3pt}){4-5} \cmidrule(l{3pt}r{3pt}){6-7}
Noise & OLS & IV & OLS & IV & OLS & IV\\
\midrule
Low & 0.1766 & 0.2477 & 0.1683 & 0.1119 & 0.1730 & 0.0868\\
Medium & 0.6785 & 0.3511 & 0.6754 & 0.1732 & 0.6819 & 0.1621\\
High & 1.2097 & 0.6023 & 1.2097 & 0.3295 & 1.2094 & 0.3323\\
\bottomrule
\end{tabular}\label{tab:rmse.sims}

\medskip

	\parbox{6.4in}{\footnotesize Notes: All other parameters are kept at their baseline values ($K=2500$ vocabulary terms, $D=100$ embedding dimensions, $n_{\text{labeled}} = 1250$). OLS RMSE is approximately invariant to $n$ within each noise stratum, consistent with a persistent bias that sampling variance alone cannot eliminate. IV RMSE declines monotonically in $n$ and is lower than OLS RMSE in all cells except the smallest sample under Low noise, where finite-sample variance in instrument-based matching temporarily outweighs the (small) OLS bias.}
\end{table}

Table~\ref{tab:rmse.sims} examines these patterns using the root mean squared error (RMSE). Consistent with Figure~\ref{fig:fig01}, RMSE for OLS is roughly flat across sample sizes within a given noise stratum (0.177/0.168/0.173 under Low noise for $n = $5K/20K/50K). This supports our conclusion of a fixed asymptotic bias that dominates any reduction in variance. The RMSE of the IV estimate, by contrast, falls sharply with $n$ (0.248 $\rightarrow$ 0.112 $\rightarrow$ 0.087 under Low noise), and the improvement is steepest exactly where OLS bias is worst (Medium and High noise specifications). The one exception is the smallest sample under Low noise, where OLS RMSE is actually lower (0.177 vs. 0.248): at the smallest sample size, the OLS bias is small, while the finite-sample variance of the nearest-neighbor instrument (whose match quality depends on the size of the instrument pool) is still relatively larger. This crossover disappears as soon as $n$ or $\sigma_\eta$ increases, which is the relevant regime for most applications with text-based predicted variables.

Table~\ref{tab:robustness.sims} tests the simulation design along seven dimensions, holding $n=20000$ fixed. 
Two of these perturbations target the assumptions underlying the instrument's validity, and both erode its performance accordingly. 

First, we add noise directly into the embedding coordinates $E_i$, rather than only into the predicted probability, as in the baseline design. Formally, we let $\widehat{X}_i = \text{logit}^{-1}( (E_i  + U_i)^\top \delta ), U_{id} \sim \mathcal{N}(0, \sigma_{\eta}^2).$ As the table shows, the RMSE rises for both estimators relative to Table~\ref{tab:rmse.sims} (IV High: 0.79 vs. 0.33 at baseline). This is because the nearest-neighbor match is computed in this same embedding space: corrupting $E_i$ also corrupts the metric used to find valid neighbors, so the instrument is less able to locate semantically close documents and its correlation with the true signal weakens, even though the noise remains idiosyncratic across documents. 

Second, and more damagingly, the noise instead perturbs the coefficient vector $\delta$, which is a shock shared across all documents in a given replication rather than one that is document-specific. Formally, we have that  $\widehat{X}_i = \text{logit}^{-1} (E_i^\top (\delta + \upsilon) ), \upsilon_d \sim \mathcal{N}(0, \sigma_{\eta}^2).$ IV RMSE nearly converges to OLS RMSE at 
Medium and High noise (0.922 vs. 1.011, and 1.425 vs. 1.502). This follows directly from the identification argument: the instrument's validity relies on the measurement error being idiosyncratic across documents, so that a document $i$ and its match $nn(i)$ have independent errors. A common shock to $\delta$ breaks this — $i$ and $nn(i)$ now share the same perturbation, the exclusion restriction fails, and the instrument inherits much of the same bias as the noisy regressor it was meant to correct.

The remaining perturbations leave both estimators' RMSE essentially unchanged from the Table~\ref{tab:rmse.sims} baseline. Varying the Gamma shape and rate parameters that govern word-frequency sparsity ($\theta$ and $\psi$, rows 3--6) has little effect on either estimator, confirming that the instrument's performance does not depend on the particular sparsity pattern of the underlying document-term matrix. Replacing the $\tanh(\cdot)$ embedding nonlinearity with ReLU produces the same result: RMSE stays comparable to baseline for both OLS and IV. Together, these robustness checks suggest the instrument is robust to how the text-generating process and embedding nonlinearity are specified, and breaks down when the noise itself violates the idiosyncratic-error assumption that identification depends on.

\begin{table}[!h]
\centering
\caption{RMSE of $\widehat{\beta}_1$ ($n=20000$) under seven alternative data-generating specifications, holding all other parameters at baseline.}
\centering
\resizebox{\ifdim\width>\linewidth\linewidth\else\width\fi}{!}{
\begin{tabular}[t]{llcccccc}
\toprule
\multicolumn{2}{c}{ } &  & & \multicolumn{2}{c}{20K} &  \\
 \cmidrule(l{3pt}r{3pt}){5-6} 
Simulation & Noise & & & OLS & IV & & \\
\cmidrule{1-8}
 & Low &  &  & 0.3904 & 0.1363 &  & \\

 & Medium &  &  & 1.3957 & 0.3531 &  & \\

\multirow[t]{-3}{*}{\raggedright\arraybackslash 1. Embedding coordinate noise} & High &  &  & 1.7012 & 0.7906 &  & \\
\cmidrule{1-8}
 & Low &  &  & 0.2244 & 0.1766 &  & \\

 & Medium &  &  & 1.0114 & 0.9222 &  & \\

\multirow[t]{-3}{*}{\raggedright\arraybackslash 2. Correlated noise across documents} & High &  &  & 1.5022 & 1.4250 &  & \\
\cmidrule{1-8}
 & Low &  &  & 0.1860 & 0.1236 &  & \\

 & Medium &  &  & 0.6648 & 0.1833 &  & \\

\multirow[t]{-3}{*}{\raggedright\arraybackslash 3. $\theta = 1.2, \psi = 0.7$} & High &  &  & 1.1926 & 0.3390 &  & \\
\cmidrule{1-8}
 & Low &  &  & 0.1759 & 0.1157 &  & \\

 & Medium &  &  & 0.6905 & 0.1855 &  & \\

\multirow[t]{-3}{*}{\raggedright\arraybackslash 4. $\theta = 1.2, \psi = 0.05$} & High &  &  & 1.2145 & 0.3534 &  & \\
\cmidrule{1-8}
 & Low &  &  & 0.1686 & 0.1154 &  & \\

 & Medium &  &  & 0.6626 & 0.1839 &  & \\

\multirow[t]{-3}{*}{\raggedright\arraybackslash 5. $\theta = 2.5, \psi = 0.3$} & High &  &  & 1.1961 & 0.3068 &  & \\
\cmidrule{1-8}
 & Low &  &  & 0.1754 & 0.1228 &  & \\

 & Medium &  &  & 0.6752 & 0.1783 &  & \\

\multirow[t]{-3}{*}{\raggedright\arraybackslash 6. $\theta = 0.8, \psi = 0.3$} & High &  &  & 1.2099 & 0.3607 &  & \\
\cmidrule{1-8}
 & Low &  &  & 0.3028 & 0.1506 &  & \\

 & Medium &  &  & 0.5618 & 0.1927 &  & \\

\multirow[t]{-3}{*}{\raggedright\arraybackslash 7. ReLU} & High &  &  & 1.0202 & 0.3514 &  & \\
\bottomrule
\end{tabular}}\label{tab:robustness.sims}

\medskip

	\parbox{6.4in}{\footnotesize Notes: The table presents seven different perturbations to the baseline simulation design. RMSE is essentially unchanged from Table~\ref{tab:rmse.sims} with changes to the word-frequency distribution (Rows 3-6). However, RMSE rises for both estimators when instrument relevance is decreased, particularly when instrument exogeneity (coefficient-level noise in Row 2) is compromised, with IV RMSE approaching that of OLS.}
\end{table}

\section{Empirical Applications}

To further demonstrate the value of our proposed approach, we revisit two recently published studies, \cite{ash} and \cite{lin2025}. We show that the algorithm can be easily adapted to each setting, allowing researchers to draw valid inferences. We briefly describe each study, relating them to our set-up, before re-analyzing them with our estimator. 

\subsection{Re-Analysis of \cite{ash}}

To show how our results can help understand research questions in political science, and how they can be applied in empirical settings, we revisit the empirical application in \cite{ash}. The paper studies reactions (e.g., applause, cheerfulness) of members of the German state parliaments (MPs) in response to speeches by their colleagues. The hypothesis is that the reactions to speeches depend on the type of speech that is being given: whether the speech is more congruent with women's topics, or men's, and that those effects depend on the speaker's gender itself. This is motivated by social identity theory and role congruity theory (e.g., \citealp{eagly}), where a speech that is framed opposite to one's gender (e.g., a male speaker on topics that are traditionally associated with women) may face a more negative reaction due to ``failing to conform to gender norms". 

To investigate whether gendered speech leads to different reactions in parliament -- and whether that depends on the congruency of the speaker's gender with the topics in the speech-- the authors collect transcripts of over 544,000 speeches across all 16 German state parliaments between 1991 and 2020. They further collect information about the MPs themselves, including gender.  Within the transcripts, they observe both the speech per se and reactions to it: e.g., applause or interjections.\footnote{As they state in the paper, stenographers in these parliaments note reactions to those speeches in those transcripts. The authors consider the seven main categories of reactions as outcomes: applause (a measure of support), cheerfulness (positive types of laughter), laughter (malicious jeering), interjections (interruptions), disagreement (verbal protest against something), agreement, and unrest (a state of noisiness or disorder).} To create a measure of gendered speech, they use a Latent Dirichlet Allocation (LDA, see \citealp{blei}) topic model to classify speeches in terms of topics and gender. Then, they regress a measure of outcome (e.g., interruptions) on this generated regressor.

Thus, their procedure falls exactly within the scope of Section \ref{main}: $Z_i$ is speech $i$ (which may be given by a man or woman), and $X_i = f(Z_i)$ is the true value of the gendered speech. However, they only observe an estimated measure (after running the LDA) given by $\hat{X}_i$. They use $\hat{X}_i$ in regression (\ref{eq3}), where $Y_i$ is a measure of support for the speech (logged reaction counts, for different types of reactions). We now reproduce their main specifications and compare them with the IV estimator that we propose. 
In particular, we split the data into three equal parts. The first split is used to train the LDA model, reproducing the results from the original paper almost exactly. We then apply the parameters learned during the training stage to each of the two remaining splits independently. One of these splits serves as the master data, while the other serves as the instrument pool. To construct a set of matches, we use nearest-neighbor matching based on a singular value decomposition (SVD) using both splits and the more than 84,000 features in the document-term matrices.\footnote{This implementation simplifies the matching because we look for matches based on 200 dimensions, rather than 84,000 features. We opt for this as an illustration of the approach, rather than fully revisiting the original paper. However, the simplification comes with two important implicit assumptions. First, we are assuming that the input $Z_i$ documents are well approximated by the SVD and that they preserve closeness in the latent variable. Formally, let $M_i$ represent the SVD for $Z_i$. Thus, we are assuming that $\|f(Z_i) - f(Z_j)\| \leq L \|M_i - M_j\|$ for some $L>0$: closneess in SVD space leads to closeness in the latent space. Second, this transformation may increase the chance of violations of the independence assumption. Indeed, for a close match, it is more likely that the approximation errors are more similar. Thus, in an empirical exercise, the researcher must be mindful when making these simplifications.} We focus on the first 200 dimensions. This dimension-reduction step is used to speed up the matching process. The results are shown in Table \ref{tab:gender_interactions_combined} below.\footnote{Table~\ref{tab:gender_interactions_combined_50} in Appendix shows that our results remain unchanged if we match based on the first 50 or 100 dimensions, respectively.}

\begin{table}[!htbp]
\centering
\caption{Reactions to Gendered Speech in the German Parliament: Original vs. IV-Corrected Estimates.}
\label{tab:gender_interactions_combined}
{
\setlength{\tabcolsep}{5pt}
\renewcommand{\arraystretch}{1.15}

\begin{tabular}{lccccc@{\hspace{3pt}}c}
\toprule\toprule
& \multicolumn{6}{c}{\textit{Predictor}} \\ 
\cmidrule(lr){2-7}
& \multicolumn{2}{c}{Female}& \multicolumn{2}{c}{Predicted Gender}& \multicolumn{2}{c}{Female $\times$ Predicted Gender} \\
\cmidrule(lr){2-3}\cmidrule(lr){4-5}\cmidrule(l){6-7}
Outcome: & Original & IV & Original & IV & Original & IV \\
\midrule

Applause
& 0.066$^{***}$ & 0.078$^{***}$
& -0.100$^{***}$ & -0.138$^{***}$
& 0.088$^{***}$ & 0.130$^{***}$ \\
& (0.009) & (0.009)
& (0.005) & (0.008)
& (0.007) & (0.011) \\[1mm]

Interjection
& 0.034$^{***}$ & 0.056$^{***}$
& -0.162$^{***}$ & -0.216$^{***}$
& 0.062$^{***}$ & 0.102$^{***}$ \\
& (0.007) & (0.008)
& (0.005) & (0.008)
& (0.007) & (0.011) \\[1mm]

Laughter
& 0.0004 & 0.0003
& -0.015$^{***}$ & -0.033$^{***}$
& 0.006$^{***}$ & 0.008$^{***}$ \\
& (0.001) & (0.001)
& (0.001) & (0.002)
& (0.001) & (0.002) \\[1mm]

Cheerfulness
& -0.010$^{***}$ & -0.004$^{**}$
& -0.023$^{***}$ & -0.037$^{***}$
& 0.009$^{***}$ & 0.019$^{***}$ \\
& (0.002) & (0.002)
& (0.001) & (0.002)
& (0.001) & (0.002) \\[1mm]

Unrest
& 0.012$^{***}$ & 0.014$^{***}$
& -0.016$^{***}$ & -0.022$^{***}$
& 0.006$^{***}$ & 0.009$^{***}$ \\
& (0.002) & (0.002)
& (0.001) & (0.002)
& (0.001) & (0.003) \\[1mm]

Disagreement
& 0.001$^{***}$ & 0.002$^{***}$
& -0.006$^{***}$ & -0.008$^{***}$
& 0.002$^{***}$ & 0.004$^{***}$ \\
& (0.001) & (0.001)
& (0.0004) & (0.001)
& (0.001) & (0.001) \\[1mm]

Agreement
& 0.002 & 0.002
& 0.001 & 0.004$^{**}$
& 0.001 & 0.001 \\
& (0.002) & (0.002)
& (0.001) & (0.002)
& (0.002) & (0.003) \\[1mm]

Positive
& 0.061$^{***}$ & 0.065$^{***}$
& -0.107$^{***}$ & -0.123$^{***}$
& 0.093$^{***}$ & 0.149$^{***}$ \\
& (0.009) & (0.009)
& (0.005) & (0.008)
& (0.007) & (0.011) \\[1mm]

Negative
& 0.039$^{***}$ & 0.063$^{***}$
& -0.173$^{***}$ & -0.236$^{***}$
& 0.068$^{***}$ & 0.110$^{***}$ \\
& (0.007) & (0.008)
& (0.005) & (0.008)
& (0.007) & (0.011) \\[1mm]

Share Positive
& 0.007$^{**}$ & 0.004
& 0.025$^{***}$ & 0.036$^{***}$
& 0.001 & -0.002 \\
& (0.003) & (0.003)
& (0.001) & (0.002)
& (0.002) & (0.003) \\[1mm]

Any
& 0.021$^{***}$ & 0.018$^{***}$
& -0.066$^{***}$ & -0.067$^{***}$
& 0.062$^{***}$ & 0.089$^{***}$ \\
& (0.006) & (0.006)
& (0.003) & (0.005)
& (0.004) & (0.006) \\

\midrule
\bottomrule
\end{tabular}
}
	\parbox{6.4in}{\footnotesize Notes: This table compares OLS estimates from \cite{ash} to IV estimates from our proposed measurement-error-correction procedure, using the same specification and outcome variables. Matches between the analysis and instrument samples are constructed via nearest-neighbor matching on an SVD of the document-term matrix (84,000+ features, first 200 dimensions retained), and the matched value of \textit{Predicted Gender} instruments for the analysis-sample value in an IV regression of each reaction-type outcome (log reaction counts) on \textit{Female}, \textit{Predicted Gender}, and their interaction. Outcomes include applause, interjections, laughter, cheerfulness, unrest, disagreement, agreement, and aggregate positive/negative/any-reaction measures. Standard errors are in parentheses.$^{***} p<0.01$, $^{**} p<0.05$, $^{*} p<0.10$}
\end{table}

Table \ref{tab:gender_interactions_combined} shows that the pattern of the original results is preserved under our IV approach, but the magnitudes are often substantially larger. Across most outcomes, the coefficient on \textit{Predicted Gender} remains negative, indicating that speeches more closely associated with women's topics tend to elicit weaker reactions. The change in magnitudes is significant: from 38\% larger for applause to 48\% for interjection, while other magnitudes have increases close to 50\% (e.g., share positive). At the same time, the interaction between \textit{Female} and \textit{Predicted Gender} remains positive and generally increases under IV, implying that this negative relationship is weaker when the speaker is a woman. This is especially clear for outcomes such as applause, interjections, and the summary measures of positive and negative reactions, where the coefficients on the interaction grow by over 60\% for applause and interjections, and by over 50\% for positive and negative reactions.

 These differences are consistent with attenuation bias in the original estimates due to measurement error in the LDA-based regressor. Once we account for this error, the negative main effect of predicted gender often becomes more negative, while the positive interaction with speaker gender becomes larger. At the same time, this is not uniform across all outcomes: for example, the estimates for agreement remain small, and the interaction for \textit{Share Positive} is not statistically significant in either specification. Overall, these results suggest that correcting for measurement error does not mechanically inflate all coefficients, but instead sharpens the role-congruity pattern precisely where it is most substantively relevant.

\subsection{Re-Analysis of \cite{lin2025}}

In the article ``Addressing Risk by Doing Good: Business Response to Government Policy Initiative,'' \cite{lin2025} argues that, in authoritarian regimes like China, firms facing greater political risk are more likely to respond actively to government policy initiatives. The author studies this through the participation of firms in China’s Targeted Poverty Alleviation (TPA) campaign after 2015. The core idea is that firms ``address risk by doing good'': by spending on a politically important state initiative, they can build ties with officials and potentially reduce future regulatory or political vulnerability. 

On the measurement side, the paper builds a firm-level Political Risk Index (PRI) from an original text dataset of 418,480 investor–firm Q\&As collected from earnings calls, IPO roadshows, analyst meetings, and related investor relations events for Chinese-listed firms from 2012 to 2019. Thus, the paper measures political risk at the firm-year level by looking at how much firm-investor discussion is devoted to politics-related concerns. The argument is that, when investors repeatedly ask about government actions, policy uncertainty, regulation, or state relations, this reveals concerns about the firm's exposure to political risk.

The first step of the measurement exercise uses BERT to classify whether each Q\&A is political or nonpolitical.\footnote{A conversation is coded as political when governments, government branches, public policies, or politicians are explicitly mentioned.} Those annotations were then used to fine-tune a Chinese BERT model.\footnote{To train the classifier, the author randomly sampled about 4,000 Q\&As, had three trained coders classify them independently, and used majority agreement as the final label.} The model identified 54,021 political conversations, or 12.9\% of all Q\&As. The resulting PRI is defined very simply as the proportion of political Q\&As among all Q\&As for a firm in a given year. The second step asks: What kind of political risk are firms talking about? Lin uses a Keyword-Assisted Topic Model (keyATM).\footnote{KeyATM allows the researcher to specify seed topics in advance, which gives more interpretable categories than an unsupervised topic model.} Based on the most frequent keywords from the earlier structural topic model, Lin defines five substantive topic groups: government finance (subsidies and tax policies), government investment and development plans, financial regulation, environmental regulation, and international relations. Each political Q\&A is then assigned to the topic with the largest topic share. 
The author then constructs topic-specific PRIs using the same logic as the original index: the number of topic-specific political Q\&As divided by all Q\&As. We consider this to be the measure $\hat{X}_i$, which is continuous and allows for mean-zero and independent errors per Assumption 1.

Table 3 in \cite{lin2025} reports estimates of a linear regression of the effect of political risk (PRI) on poverty-alleviation expenditure (using a difference-in-differences design). The treatment setting is defined by the launch of the TPA campaign in 2015. Before then, firms did not face the same political opportunity to align themselves with the state through poverty-alleviation spending. The main variable of interest is the interaction between PRI and Post2015.

To implement our proposed approach, we randomly divide the data into two equal groups: an instrument group and an analysis (training) group. For each firm-year observation in the training group, we identify a suitable match among the observations in the instrument group using Euclidean distance computed from the topic-specific PRIs.\footnote{Unlike in the previous application and the simulations, we do not have access to the source DTMs. As a result, the matching procedure relies on a function of the textual data--the topic-specific PRIs--rather than on the textual data themselves. We do not expect this departure to create a major complication because documents with similar topic distributions under keyATM should also have similar underlying word-frequency profiles.} We then estimate the effect using IV, where the matched PRI serves as an instrument for the firm's PRI. 

The results in Table~\ref{tab:main_results_lin} show that the interaction between firm-level political risk and the post-2015 period is positive and statistically significant, supporting the claim that firms respond strategically when a policy initiative creates an opportunity for political alignment. The original estimates, however, are smaller in magnitude than the corresponding IV estimates, suggesting a bias due to measurement error in the main explanatory variable. Substantively, a 1-percentage-point increase in firm-level political risk is associated with about a 2.1\% increase (from a specification with year and firm fixed effects and firm-level controls) in poverty-alleviation expenditure after 2015, which corresponds to roughly \$10,600 on average (a figure 30\% larger than the original finding). 

\begin{table}[!htbp]\centering
\caption{Effect of Firm-Level Political Risk on Poverty-Alleviation Expenditure: Original vs. IV-Corrected Estimates.}
\label{tab:main_results_lin}
\begin{threeparttable}
\setlength{\tabcolsep}{7pt}
\renewcommand{\arraystretch}{1.15}
\begin{tabular}{lcccccc}
\toprule\toprule
& \multicolumn{6}{c}{\textit{Dependent variable: Expenditure}} \\
\cmidrule(lr){2-7}
& \multicolumn{2}{c}{Baseline} & \multicolumn{2}{c}{Fixed Effects} & \multicolumn{2}{c}{Controls} \\
\cmidrule(lr){2-3} \cmidrule(lr){4-5} \cmidrule(lr){6-7}
& Original & IV & Original & IV & Original & IV \\
\midrule
\textit{Variables} \\
   PRI                     & 0.0010        & 0.002$^{*}$   & -0.009$^{***}$ & -0.006        & -0.007$^{***}$ & -0.003\\   
                           & (0.0008)      & (0.001)       & (0.002)        & (0.006)       & (0.002)        & (0.006)\\   
   Post 2015               & 0.590$^{***}$ & 0.491$^{***}$ &                &               &                &   \\   
                           & (0.059)       & (0.103)       &                &               &                &   \\   
   PRI $\times$ Post 2015  & 0.019$^{***}$ & 0.025$^{***}$ & 0.019$^{***}$  & 0.026$^{***}$ & 0.016$^{***}$  & 0.021$^{***}$\\   
                           & (0.003)       & (0.006)       & (0.004)        & (0.007)       & (0.003)        & (0.007)\\ 
   \midrule
\textit{Fixed effects} \\
Firm 
  & No & No & Yes & Yes & Yes & Yes \\
Year 
  & No & No & Yes & Yes & Yes & Yes \\
  \midrule
Controls 
  & No & No & No & No & Yes & Yes \\
\midrule
\textit{Fit statistics} \\
Observations 
  & 9,466 & 4,652 & 9,466 & 4,652 & 9,371 & 4,612 \\
$R^2$ 
  & 0.082 & 0.073 & 0.555 & 0.633 & 0.542 & 0.630 \\
Within $R^2$ 
  &       &       & 0.010 & 0.003 & 0.013 & 0.009 \\
\bottomrule\bottomrule
\end{tabular}
\end{threeparttable}

\medskip

	\parbox{6.4in}{\footnotesize Notes: This table compares OLS difference-in-differences estimates from \cite{lin2025} to IV estimates from our proposed correction procedure, applied to the same firm-year panel of Chinese listed firms (2012--2019). To implement the IV correction, firm-year observations are randomly split into an analysis group and an instrument group, and each analysis-group observation is matched to an instrument-group observation via nearest-neighbor matching on Euclidean distance in topic-specific PRIs; the matched PRI instruments the firm's own PRI (and its interaction with `Post 2015') in IV. The outcome is poverty-alleviation expenditure. Clustered standard errors at the firm level are in parentheses. $^{***} p<0.01$, $^{**} p<0.05$, $^{*} p<0.10$.}

\end{table}

\FloatBarrier

\section{Discussion}

Prediction-generated measures are becoming increasingly popular in empirical research in political science and the social sciences. Researchers routinely use large language models, topic models, classifiers, and other machine-learning methods to measure concepts that are difficult to observe directly, including sentiment, political risk, media slant, policy positions, and gendered language. These measures are often subsequently used as explanatory variables in regression models, even though they are estimated with error. This paper shows that treating such variables as if they were observed without error can produce substantial bias in downstream estimates. To address this problem, we develop a simple instrumental-variables approach based on sample splitting. By generating multiple measures of the same underlying concept from independent subsamples and using one measure as an instrument for another, researchers can isolate variation associated with the latent construct from variation caused by prediction error. The procedure is computationally simple and does not require researchers to collect new data, manually label an additional validation sample, or impose a fully parametric model for the measurement-error process.

We establish the conditions under which the split-sample instrument is both relevant and exogenous and provide a practical algorithm for implementing the estimator when identical observations are not available across the independent samples. In particular, we use nearest-neighbor matching to identify observations that are similar in the feature or embedding space and then use the prediction associated with the matched observation as an instrument for the prediction in the analysis sample. Our simulations demonstrate the consequences of ignoring prediction error and the potential gains from the proposed correction. Across a range of sample sizes, noise levels, document-term distributions, and embedding specifications, OLS estimates based on noisy generated regressors remain biased even as the sample size increases. In contrast, the IV estimates are generally centered near the true parameter and become more precise as the sample grows. The simulations also clarify the limits of the method. Its performance deteriorates when matching becomes less informative or when prediction errors are correlated across observations, as occurs when all predictions share a common perturbation. These results emphasize that independent sample splitting is not merely a computational device: it is central to the identification strategy because it helps separate the measurement errors contained in the alternative measures.

The two empirical reanalyses illustrate that these biases can meaningfully affect substantive conclusions. Together, our results show that prediction error need not merely add noise to an analysis; it can systematically distort the magnitude of relationships that researchers seek to estimate. The approach developed here provides a practical correction when independent measures can be constructed, although its validity depends on the independence of their measurement errors and the strength of the relationship between matched observations. Future work could examine alternative matching and aggregation procedures, the use of more than two sample splits or multiple prediction models, and extensions to nonlinear outcome models and more complex forms of dependence. More broadly, the paper highlights that researchers should treat prediction-generated regressors as estimated quantities rather than observed covariates and should incorporate the uncertainty created during measurement into downstream statistical inference.

\bibliographystyle{plainnat}
\bibliography{gptbib}

\newpage
\appendix

\Large
\begin{center}
\textsc{Online Appendix to ``When Predictions Become Regressors: A Split-Sample Correction for Biases in Downstream Inference"}
\end{center}

\normalsize
\section{Proofs}\label{proofs}

\begin{lemma}\label{lemma1}
Under Assumption 1, if $\hat{X}_i$ is used instead of $X_i$ in the regression (\ref{eq1}), then
the OLS bias need not take the classical attenuation bias form.
\end{lemma}

\begin{proof}[Proof of Lemma \ref{lemma1}]

Consider the OLS estimator for $\beta_1$ in (\ref{eq3}), $\hat{\beta}_1 = \frac{\sum_{i=1}^n(\hat{X}_i - \overline{\hat{X}}) Y_i}{\sum_{i=1}^n (\hat{X}_i - \bar{\hat{X}})^2}$. This can be re-written as:
\begin{eqnarray*}
\hat{\beta}_1 &=& \frac{\sum_{i=1}^n(\hat{X}_i - \overline{\hat{X}_i})(\beta_0 + \beta_1 \hat{X}_i + u_i)}{\sum_{i=1}^n \left(\hat{X}_i - \overline{\hat{X}_i}\right)^2} = 0 + \beta_1 + \frac{\sum_{i=1}^n(\hat{X}_i - \overline{\hat{X}_i}) (\varepsilon_i-\beta_1 \eta_i)}{\sum_{i=1}^n \left(\hat{X}_i - \overline{\hat{X}_i}\right)^2}.
\end{eqnarray*}
Taking the probability limits on both sides as $n \to \infty$:
\begin{eqnarray}
plim_{n \to \infty} \hat{\beta}_1 &=& \beta_1 + \frac{Cov(\hat{X}_i, \varepsilon_i)}{Var(\hat{X}_i)} - \beta_1\frac{Cov(\hat{X}_i, \eta_i)}{Var(\hat{X}_i)}.\label{eqapp}
\end{eqnarray}
As $\eta_i$ is correlated with $\hat{X}_i$ by (\ref{eq2}), the last term is non-zero. Now, $Cov(\hat{X}_i, \eta_i) = Cov(X_i, \eta_i) + Var(\eta_i) = Var(\eta_i)$ by equation (\ref{eq2}) and by using $\mathbb{E}[\eta_i \mid X_i] = 0$. 
However, $Cov(\hat{X}_i, \varepsilon_i) = Cov(X_i, \varepsilon_i) + Cov(\eta_i, \varepsilon_i) = 0 + Cov(\eta_i, \varepsilon_i)$ need not equal 0, as we have made no assumptions on the latter term. In particular, the latter term can be negative if $\beta_1>0$ and $Cov(\eta_i, \varepsilon_i)<-\beta_1 Var(X_i)$.
\end{proof}

\begin{proof}[Proof of Proposition \ref{proposition}]
Recall that the endogeneity comes from $Cov(\hat{X}_i, u_i) \neq 0$. Now, consider that $u_i$ comes from the training sample, and $\hat{X}_i^{instr}$ from the instrument one. Since $Z_{nn(i)} = Z_i$ implies $f(Z_{nn(i)}) = X_i$, under the (perfect) match, we have that $\hat{X}_i^{instr} = X_i + \eta_{nn(i)}$ (i.e., it also measures $X_i$, but with a different measurement error). By random splits being disjoint, $nn(i) \neq i$. Then,
\begin{eqnarray*}
Cov(\hat{X}_i^{instr}, u_i)&=&Cov(\hat{X}_i^{instr}, \varepsilon_i - \beta_1\eta_i) \\
&=& Cov(X_i + \eta_{nn(i)}, \varepsilon_i) - \beta_1 Cov(X_i + \eta_{nn(i)},\eta_i)\\
&=&  Cov(X_i, \varepsilon_i) + Cov(\eta_{nn(i)}, \varepsilon_i) - \beta_1 Cov(X_i, \eta_i) - \beta_1 Cov(\eta_{nn(i)}, \eta_i)\\
&=& 0,
\end{eqnarray*}
where the final result comes from each of the four terms being equal to 0. Indeed, we use that $Cov(X_i, \varepsilon_i) = 0$ by exogeneity of $\varepsilon_i$ relative to $X_i$ (first term), $Cov(X_i, \eta_i) = 0$ as $\mathbb{E}[\eta_i \mid X_i] = 0$ (since $X_i = f(Z_i)$, and $f(\cdot)$ is $Z_i$-measurable, for the third term). 

As for the second and fourth terms, we have to account for the randomness in the choice of $nn(i)$. Let $U$ be the randomization device for Assumption 2. Note that, we can write the index as $\sum_{j \in instr} 1_{nn(i)=j} \eta_j$ so that $\mathbb{E}[\eta_{nn(i)} \mid Z_1,...,Z_n, U] =\sum_{j \in instr} 1_{nn(i)=j} \mathbb{E}[\eta_j \mid Z_j] = 0$ where the first equality uses Assumption 2, and the last equality uses Assumption 1. Finally, the second term equals 0 as $\mathbb{E}\eta_{nn(i)} \varepsilon_i = \mathbb{E}\left(\mathbb{E} [\eta_{nn(i)} \mid Z_1,...,Z_n, U]\mathbb{E}[\varepsilon_i \mid Z_1,...,Z_n, U]\right) = 0$ where the factorization follows from conditional independence of $\eta_{nn(i)}, \varepsilon_i$ with $nn(i) \neq i$. The fourth term follows analogously by replacing $\varepsilon_i$ by $\eta_i$. Thus, the instrument is exogenous.

As for relevance, note that $\hat{X}_i$ and $\hat{X}_i^{instr}$ are correlated as long as they measure the same underlying concept. Indeed, for $nn(i) \neq i$, but with $Z_{nn(i)} = Z_i$:
\begin{eqnarray*}
Cov(\hat{X}_i, \hat{X}_i^{instr})&=& Cov(X_i + \eta_i, X_i + \eta_{nn(i)}) \\
&=& Cov(X_i, X_i) + Cov(\eta_i, X_i) + Cov(X_i, \eta_{nn(i)}) + Cov(\eta_i, \eta_{nn(i)}) \\
&=&Var(X_i)>0,
\end{eqnarray*}
since $Cov(X_i,X_i) = Var(X_i)>0$, the other terms are 0 (by Assumption 1 and by the same arguments in the proof of exogeneity), and where we use that both $\hat{X}_i$ and $\hat{X}_i^{instr}$ are prediction-generated regressors of $X_i$ from equation (\ref{eq2}). 
\end{proof}

\begin{proposition}\label{prop_imperfect}
Let $(\hat{X}_i, \hat{X}_i^{instr})$ be two separate prediction-generated regressors computed from independent splits of the data, with $Z_i$ not necessarily equal to $Z_{nn(i)}$. 

Then, under Assumptions 1-2 and $\mathbb{E}[\varepsilon_i \mid Z_i] = 0$, $Cov(\hat{X}_i, \hat{X}_i^{instr}) \neq 0$ and $Cov(u_i, \hat{X}_i^{instr}) = 0$, even under imperfect matching.
\end{proposition}

\begin{proof}[Proof of Proposition \ref{prop_imperfect}]
Following the proof of Proposition \ref{proposition}, we have that:
\begin{eqnarray*}
Cov(\hat{X}_i^{instr}, u_i)&=&Cov(\hat{X}_i^{instr}, \varepsilon_i - \beta_1\eta_i) \\
&=& Cov(X_{nn(i)} + \eta_{nn(i)}, \varepsilon_i) - \beta_1 Cov(X_{nn(i)} + \eta_{nn(i)},\eta_i).
\end{eqnarray*}
The same steps of the proof of Proposition \ref{proposition} imply that the second and fourth terms are equal to 0. The first term is 0 because $\mathbb{E}[\varepsilon_i \mid Z_1,...,Z_n, U] = \mathbb{E}[\varepsilon_i \mid Z_i] = 0$ from the strengthened assumption and Assumption 1, and $f(Z_{nn(i)})$ is measurable. The third term is equal to zero since $\mathbb{E}[X_{nn(i)} \eta_i \mid Z_1,...,Z_n, U] = f(Z_{nn(i)}) \mathbb{E}[\eta_i \mid Z_i] = 0$. 

As for relevance,
\begin{eqnarray*}
Cov(\hat{X}_i, \hat{X}_i^{instr}) &=& Cov(X_i+\eta_i, X_{nn(i)} + \eta_{nn(i)}) \\
&=& Cov(X_i, X_{nn(i)}) \\
&=& Cov(X_i, X_i + \Delta_i) \\
&=& Var(X_i) + Cov(X_i, \Delta_i).
\end{eqnarray*}
where the first line follows the proof of exogeneity above, and we set $\Delta_i \equiv f(Z_{nn(i)})-f(Z_i)$. Thus, if and only if $-Cov(X_i,\Delta_i) \neq Var(X_i) $, we have relevance. 
\end{proof}

\section{Additional Table Referenced in the Text}

\begin{table}[!htbp]
\centering
\caption{\small Reactions to Gendered Speech in the German Parliament:
Original vs.\ IV-Corrected Estimates. This table compares OLS estimates
from \cite{ash} with IV estimates from our proposed measurement-error
correction procedure. Matches are constructed using nearest-neighbor
matching on 50- and 100-dimensional SVD representations of the
document-term matrix. Standard errors are in parentheses.
$^{***}p<0.01$, $^{**}p<0.05$, $^{*}p<0.10$.}
\label{tab:gender_interactions_combined_50}

\begingroup
\scriptsize
\setlength{\tabcolsep}{2.5pt}
\renewcommand{\arraystretch}{0.95}

\begin{tabular}{lccccccccc}
\toprule\toprule
& \multicolumn{9}{c}{\textit{Predictor}} \\
\cmidrule(lr){2-10}
& \multicolumn{3}{c}{Female}
& \multicolumn{3}{c}{Predicted Gender}
& \multicolumn{3}{c}{Female $\times$ Predicted Gender} \\
\cmidrule(lr){2-4}
\cmidrule(lr){5-7}
\cmidrule(l){8-10}
Outcome
&  & \multicolumn{2}{c}{IV} 
&  & \multicolumn{2}{c}{IV}
& & \multicolumn{2}{c}{IV} \\ \cmidrule(lr){3-4} \cmidrule(lr){6-7} \cmidrule(lr){9-10}
& Original & 50 dim. & 100 dim.
& Original & 50 dim. & 100 dim.
& Original & 50 dim. & 100 dim. \\
\midrule

Applause
& 0.066$^{***}$ & 0.083$^{***}$ & 0.078$^{***}$
& -0.100$^{***}$ & -0.149$^{***}$ & -0.141$^{***}$
& 0.088$^{***}$ & 0.136$^{***}$ & 0.136$^{***}$ \\
& (0.009) & (0.009) & (0.009)
& (0.005) & (0.008) & (0.008)
& (0.007) & (0.011) & (0.011) \\[0.5mm]

Interjection
& 0.034$^{***}$ & 0.062$^{***}$ & 0.058$^{***}$
& -0.162$^{***}$ & -0.230$^{***}$ & -0.222$^{***}$
& 0.062$^{***}$ & 0.110$^{***}$ & 0.116$^{***}$ \\
& (0.007) & (0.008) & (0.008)
& (0.005) & (0.008) & (0.008)
& (0.007) & (0.011) & (0.011) \\[0.5mm]

Laughter
& 0.0004 & 0.0003 & 0.0003
& -0.015$^{***}$ & -0.016$^{***}$ & -0.016$^{***}$
& 0.006$^{***}$ & 0.005$^{***}$ & 0.005$^{**}$ \\
& (0.001) & (0.001) & (0.001)
& (0.001) & (0.002) & (0.002)
& (0.001) & (0.002) & (0.002) \\[0.5mm]

Cheerfulness
& -0.010$^{***}$ & -0.001 & -0.003
& -0.023$^{***}$ & -0.037$^{***}$ & -0.034$^{***}$
& 0.009$^{***}$ & 0.018$^{***}$ & 0.015$^{***}$ \\
& (0.002) & (0.002) & (0.002)
& (0.001) & (0.002) & (0.002)
& (0.001) & (0.003) & (0.003) \\[0.5mm]

Unrest
& 0.012$^{***}$ & 0.017$^{***}$ & 0.016$^{***}$
& -0.016$^{***}$ & -0.024$^{***}$ & -0.020$^{***}$
& 0.006$^{***}$ & 0.009$^{***}$ & 0.010$^{***}$ \\
& (0.002) & (0.002) & (0.002)
& (0.001) & (0.002) & (0.002)
& (0.001) & (0.003) & (0.003) \\[0.5mm]

Disagreement
& 0.001$^{***}$ & 0.002$^{***}$ & 0.002$^{***}$
& -0.006$^{***}$ & -0.007$^{***}$ & -0.008$^{***}$
& 0.002$^{***}$ & 0.002$^{*}$ & 0.005$^{***}$ \\
& (0.001) & (0.001) & (0.001)
& (0.0004) & (0.001) & (0.001)
& (0.001) & (0.001) & (0.001) \\[0.5mm]

Agreement
& 0.002 & -0.001 & -0.001
& 0.001 & 0.003 & 0.003
& 0.001 & -0.0001 & 0.0002 \\
& (0.002) & (0.002) & (0.002)
& (0.001) & (0.002) & (0.002)
& (0.002) & (0.003) & (0.003) \\[0.5mm]

Positive
& 0.061$^{***}$ & 0.079$^{***}$ & 0.075$^{***}$
& -0.107$^{***}$ & -0.161$^{***}$ & -0.152$^{***}$
& 0.093$^{***}$ & 0.145$^{***}$ & 0.144$^{***}$ \\
& (0.009) & (0.010) & (0.010)
& (0.005) & (0.008) & (0.008)
& (0.007) & (0.011) & (0.011) \\[0.5mm]

Negative
& 0.039$^{***}$ & 0.069$^{***}$ & 0.064$^{***}$
& -0.173$^{***}$ & -0.243$^{***}$ & -0.234$^{***}$
& 0.068$^{***}$ & 0.116$^{***}$ & 0.124$^{***}$ \\
& (0.007) & (0.008) & (0.008)
& (0.005) & (0.008) & (0.008)
& (0.007) & (0.011) & (0.012) \\[0.5mm]

Share Positive
& 0.007$^{**}$ & 0.005 & 0.007$^{*}$
& 0.025$^{***}$ & 0.027$^{***}$ & 0.025$^{***}$
& 0.001 & 0.002 & 0.0002 \\
& (0.003) & (0.004) & (0.004)
& (0.001) & (0.002) & (0.002)
& (0.002) & (0.003) & (0.003) \\[0.5mm]

Any
& 0.021$^{***}$ & 0.041$^{***}$ & 0.035$^{***}$
& -0.066$^{***}$ & -0.114$^{***}$ & -0.102$^{***}$
& 0.062$^{***}$ & 0.096$^{***}$ & 0.091$^{***}$ \\
& (0.006) & (0.006) & (0.006)
& (0.003) & (0.005) & (0.005)
& (0.004) & (0.007) & (0.007) \\

\bottomrule
\end{tabular}
\endgroup
\end{table}

\end{document}